\documentclass[10pt]{IEEEtran}
\ifCLASSINFOpdf
\else
\fi
\usepackage{subfig}
\usepackage{array}
\usepackage{subfig}
\usepackage{amsmath}
\usepackage[nolist,nohyperlinks]{acronym}
\usepackage{graphicx,wrapfig,lipsum}
\usepackage{rotating}
\usepackage{amsmath, amssymb, amsthm}
\usepackage{stmaryrd}
\usepackage{tikz}
\usepackage{verbatim}
\usepackage{url}
\usepackage[utf8]{inputenc}
\usepackage[english]{babel}
\newtheorem{theorem}{Theorem}[]

\newtheorem{lemma}[]{Lemma}
\usepackage{multicol}
\usepackage{xr-hyper}
\usepackage{algorithm}
\usepackage{algpseudocode}
\usepackage{scalerel}
\usepackage{multicol}
\usepackage{setspace}
\newtheorem{definition}{Definition}
\usepackage[noadjust]{cite}
\usepackage{bbm}
\usepackage[normalem]{ulem}
\usepackage{color}
\usepackage{dsfont}
\usepackage{flushend}
\usetikzlibrary{automata}
\usetikzlibrary{arrows}
\newcommand\blfootnote[1]{%
  \begingroup
  \renewcommand\thefootnote{}\footnote{#1}%
  \addtocounter{footnote}{-1}%
  \endgroup
}
\allowdisplaybreaks
\begin{document}
%\bstctlcite{IEEEexample:BSTcontrol}
	\title{Accurate Characterization of Dynamic Cell Load in Noise-Limited Random Cellular Networks}
	\author{Gourab Ghatak$^{\dagger}$ $^\ddagger$, Antonio De Domenico$^{\dagger}$, and Marceau Coupechoux$^\ddagger$
 \\ \small{ $^{\dagger}$CEA, LETI, MINATEC, F-38054 Grenoble,
France; $^\ddagger$LTCI, Telecom ParisTech, Universit\'e Paris Saclay, France.}
\\ \small{Email: gourab.ghatak@cea.fr; antonio.de-domenico@cea.fr, and marceau.coupechoux@telecom-paristech.fr}\vspace*{-0.7cm}}
		\maketitle
        \thispagestyle{empty}
	\begin{abstract}
    \blfootnote{The research leading to these
results are jointly funded by the European Commission (EC) H2020 and the Ministry of Internal affairs and Communications (MIC) in Japan under grant agreements N$^o$ 723171 5G MiEdge.}
The analyses of cellular network performance based on stochastic geometry generally ignore the traffic dynamics in the network. This restricts the proper evaluation and dimensioning of the network from the perspective of a mobile operator. To address the effect of dynamic traffic, recently, the mean cell approach has been introduced, which approximates the average network load by the zero cell load. However, this is not a realistic characterization of the network load, since a zero cell is statistically larger than a random cell drawn from the population of cells, i.e., a typical cell. In this paper, we analyze the load of a noise-limited network characterized by high signal to noise ratio (SNR). The noise-limited assumption can be applied to a variety of scenarios, e.g., millimeter wave networks with efficient interference management mechanisms. First, we provide an analytical framework to obtain the cumulative density function of the load of the typical cell. Then, we obtain two approximations of the average load of the typical cell. We show that our study provides a more realistic characterization of the average load of the network as compared to the mean cell approach. Moreover, the prescribed closed-form approximation is more tractable than the mean cell approach.

%In this paper, we characterize a single tier millimeter wave (mm-wave) cellular network in terms of the cell load of the typical cell under dynamic traffic using elements of queuing theory and stochastic geometry. First, we provide an analytical framework to obtain the cumulative density function of the cell load. Furthermore, we obtain one single-integral expression, and one closed-form equation, to approximate the mean load of the typical cell. On the one hand, the proposed closed-form expression is an excellent approximation for the mean cell load calculated using the mean cell approach; on the other hand, it is a more realistic way of modeling the cell load as we utilize the distributions of the typical cell rather than the signal to noise ratio distribution of the typical user.
	\end{abstract}
	\IEEEpeerreviewmaketitle
            \vspace*{-0.3cm}
	\section{Introduction}
      \vspace*{-0.10cm}
    Stochastic geometry has emerged as an important tool for modeling and analyzing large scale wireless cellular networks~\cite{elsawy2013stochastic}, wherein the performance is typically characterized by studying metrics such as \ac{SINR} coverage probability and user throughput. To effectively model the user throughput and to efficiently dimension a cellular network from the operators' perspective, the characterization of the cell load is necessary. %However, existing stochastic geometry studies either ignore the traffic dynamics or assume saturated traffic (see e.g.,~\cite{7493676}) to maintain tractability. With this assumption, the load of a cell is approximated by using only the average number of its associated users.     
   The existing literature in stochastic geometry models the cell load by considering the average number of associated full buffer users, uniformly distributed over the cell area, see e.g.,~\cite{elsawy2013stochastic, 7493676}. This is not realistic since it ignores the effect of dynamic traffic on the user distribution: users with low data rate tend to stay longer in the system, and as a result, the user distribution becomes inhomogeneous in space.  
   %For a more realistic characterization, a dynamic traffic model should be considered.
   
   However, studying dynamic traffic using elements of queuing theory in stochastic geometry based analyses is still an open problem~\cite{elsawy2013stochastic}. In this regard, Blaszczyszyn et al.~\cite{blaszczyszyn2014user}, have introduced the mean cell approach %A spatial birth-death process was very recently considered by Sankaraman and Baccelli using Poisson-dipoles to model D2D networks~\cite{sankararaman2017spatial}.
   %However, studying dynamic traffic using elements of queuing theory in stochastic geometry based analyses is generally non-trivial. To deal with this challenge, recently, the mean cell approach has been introduced by Blaszczyszyn et. al.~\cite{blaszczyszyn2014user} to study a single tier network. 
which avoids extensive simulations by approximating the spatial \ac{SINR} distribution of a cell with the \ac{SINR} distribution of the typical user. Thus, in essence, the mean cell approach characterizes the load of the cell containing the typical user, i.e., the \textit{zero cell}~\cite{chiu2013stochastic}. Although this approach enables modeling the cell load, it may lead to intractable expressions consisting of multiple integrals for evaluation of the \ac{SINR} coverage probability. Moreover, a characterization of the load of the zero cell is not a reliable metric for evaluating the network wide load distribution since the zero cell is statistically larger than a random cell drawn from the population of cells, i.e., a \textit{typical cell}. To understand this intuitively, one can assume a random sample point and select the cell containing the point. By stationarity,
the distribution of this cell coincides with that of the zero cell. Since the sample point tends to fall with greater probability into larger cells, the zero cell tends to be larger than the
typical cell.
In this paper, for the case of noise-limited networks, we provide approximations for the network load by characterizing the load of the typical cell. This provides a more realistic characterization of the network load. This noise-limited assumption can be applied to a variety of contexts. For example, in millimeter wave (mm-wave) networks that utilize directional antennas and advanced interference management mechanisms, the performance tends to be noise-limited. Singh et.al~\cite{5733382} have shown the validity of the noise-limited network assumption in mm-wave mesh networks. Furthermore, this noise-limited scenario enables us to visualize our results in light of the seminal work of Bonald et al.~\cite{bonald2003wireless} who derived the cell load expressions for a single cell with dynamic traffic.

The contribution of this paper is as follows. We obtain a closed-form expression for the \ac{CDF} of the load of the typical cell in a noise-limited network by considering dynamic traffic. We use it to obtain the fraction of stable cells for a given deployment density of small cells. %Moreover, we provide the required antenna gain to guarantee network stability for a given deployment density. 
Then, we obtain a single integral-based approximation, and a closed-form expression, for the average load of the typical cell. We show that the first approximation models the cell load from a network perspective more accurately than the mean cell approach. Whereas, the closed-form expression provides a faster and more tractable alternative to calculate the network load, since it does not require evaluation of integrals.

 The rest of the paper is organized as follows. In Section \ref{sec:SM} we introduce the single tier network and the associated parameters. In Section \ref{sec:MR}, we present our main results on the \ac{CDF} and the average of the load of the typical cell of the network. In Section \ref{sec:SR}, we present the results on the stable fraction of the network and we show the accuracy of our derived approximations with respect to Monte-Carlo simulations. Finally, the paper concludes in Section \ref{sec:C}.
    \vspace*{-0.2cm}
\section{System Model}
  \vspace*{-0.05cm}
\label{sec:SM}
We consider a single-tier cellular network equipped with advanced interference management algorithms, so that the user performance is noise-limited. %This allows us to compare our results with the classical cell load model by Bonald et.al.~\cite{bonald2003wireless}.

The positions of the \ac{BS} are modeled as points of a \ac{PPP} $\phi$ with intensity $\lambda$ [m$^{-2}$].
The \ac{BS}s operate with a transmit power $P_t$, and the product of the gains of the antennas at the transmitter and the receiver is $G_0$. We consider a fast fading that is Rayleigh distributed with variance equal to one. Furthermore, we assume a path loss model where the power at the origin received from a \ac{BS} located at a distance $r$ is given by $P_r = K\cdot P_t \cdot h \cdot G_0 \cdot r^{-\alpha}$, where $K$ is the path loss coefficient, $h$ is the exponentially distributed fading power, and $\alpha$ is the path loss exponent. Thus, the average SNR can be written as $\frac{K\cdot P_t \cdot G_0 \cdot r^{-\alpha}}{N_0 \cdot B}$ =  $\xi r^{-\alpha}$, where $\xi = \frac{K\cdot P_t \cdot G_0}{N_0 \cdot B}$ is the average SNR at $1~\text{m}$. $N_0$ and $B$ are the noise power density and the operating bandwidth, respectively.

In this network, we assume that the users arrive in the system, download a file, and leave the system. Any new download by the same user is considered as a new user. The arrival process of the new users is Poisson distributed with an intensity $\lambda_U$ [users $\cdot$s$^{-1}\cdot$m$^{-2}$] and these new users are uniformly distributed over the network area $A$. The average file size is $\sigma$ [bits/user]. When there are $n$ users simultaneously served by a BS, the available resources are equally shared between them in a Round Robin fashion. Accordingly, we define the traffic density $w$ in the network as $w = \lambda_U \cdot \sigma$ [bits$\cdot$s$^{-1}\cdot$m$^{-2}$]. Note that, while the user arrivals are uniform in space, as the space-time process evolves, users farther from the serving BSs, i.e., characterized by lower data rates, stay longer in the system, resulting in an inhomogeneous distribution of active users in the network.
 \vspace*{-0.05cm}
\section{Characterization of the Network Load}
 \vspace*{-0.1cm}
\label{sec:MR}
\subsection{Static vs Dynamic Load}
 \vspace*{-0.1cm}
Before proceeding to our main results, it is necessary to discuss the distinction of dynamic cell-load characterization as compared to the approaches that model the cell load as simply the average number of associated users to a BS~\cite{singh2014joint}. For this we compare the downlink user throughput using the two approaches to study the difference.
First, in our system model, we calculate the dynamic cell load using the study of Bonald et.al.~\cite{bonald2003wireless}, say $\bar{\rho}$. Then use it to calculate the downlink user throughput $R_{dyn}$, which is given by~\cite{bonald2003wireless}:
\begin{align}
R_{dyn} = w \frac{1 - \bar{\rho}}{\bar{\rho}} \cdot A,
\end{align}
where, the area of the typical cell $A$ is approximated as $A = \frac{1}{\lambda}$, and the average number of active users in the dynamic traffic model~\cite{blaszczyszyn2014user} is:
$
N = \frac{\bar{\rho}}{1 - \bar{\rho}}.
$
To compare $R_{dyn}$ with that obtained using the analyses in the existing literature~\cite{7493676,singh2014joint}, we assume that the users are located homogeneously in each cell of the network following a PPP with an average of $N$ users per cell. Using this assumption, we carry out simulations to obtain the downlink user throughput $T_{PPP}$ as:
\begin{align}
R_{PPP} = \mathbb{E}_n \left[\frac{B}{n} \log_2\left(1 + SINR\right)\right],
\end{align}
where the expectation is taken with respect to the number of users in each cell ($n$), which is Poisson distributed with mean $N$, as well as the SINR of each user in each cell. The difference between $R_{dyn}$ and $R_{PPP}$ is highlighted in Fig. \ref{fig:N_user}. Even though the average number of users in both cases are same, as the space-time process with dynamic traffic evolves, the user distribution is no longer homogeneous in space which is not taken into account in existing studies.
%Therefore, the throughput of the typical user calculated by considering dynamic traffic $R_{dyn}$ is actually lower than the one computed with the classic approach $R_{PPP}$.
\begin{figure}
\centering
\includegraphics[width = 7 cm, height =3.4 cm]{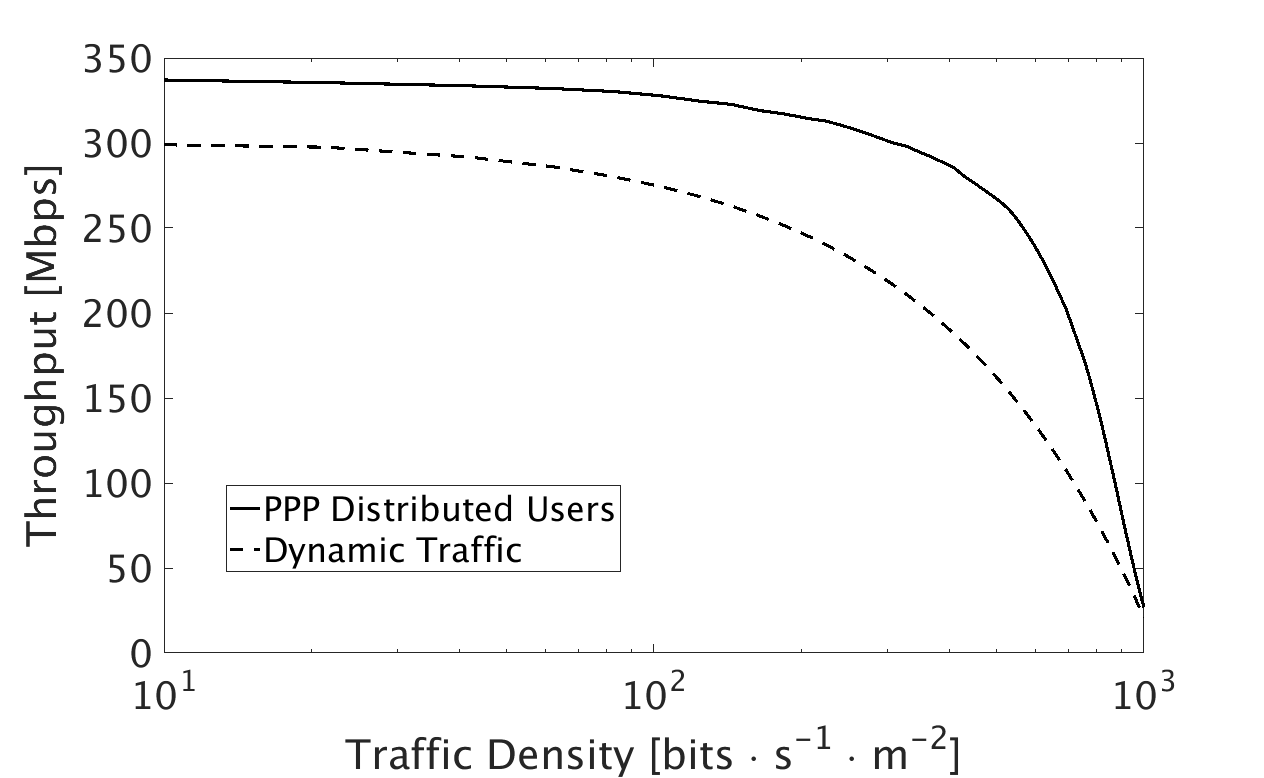}
\caption{Throughput comparison of mean cell approach with PPP distributed users.}
\label{fig:N_user}
\vspace*{-0.5cm}
\end{figure}
  \vspace*{-0.3cm}
\subsection{Preliminaries}
In case of single-tier random cellular networks, the cell of a \ac{BS} is given by the \ac{PV} partition of the space~\cite{chiu2013stochastic}. In the $\mathbb{R}^2$ plane, the \ac{PV} region of a \ac{BS} located at $x_0 \in \phi$ is:
%\begin{align*}
$\mathcal{A} = \{y : ||y - x_0|| < ||y - x_i||; \forall x_i\in \phi\backslash \{x_0\} \}.$ %\end{align*}

The mosaic of the cells formed for all such $x_0$ from a PPP is called a \ac{PV} network.
To investigate the geometry-dependent characteristics of the cells (e.g., the cell load), in a stationary random \ac{PV} network, it is imperative to define the notion of the ‘average' cell. Thus, we recall the following definitions that provide a characterization of the average cell.

\begin{definition} The zero cell or the Crofton cell of the \ac{PV} network is defined as the cell containing a given fixed point in its interior~\cite{chiu2013stochastic}. 
\end{definition}
%By the stationarity of the \ac{PV} network, the resulting random shape of the cell does not depend on the choice of the fixed point, hence the origin can be assumed to be this point.
\begin{definition} The typical cell of a \ac{PV} network is defined as a cell selected at random within a large region of the network with equal chances for each cell to be picked~\cite{chiu2013stochastic}. Thereafter, the network is translated so that the center of the typical cell becomes the origin.
\end{definition}
The zero cell versus typical-cell approaches of modeling the network performance corresponds to the evaluation from the perspectives of the user and the network operator, respectively. Hence, for the case of analyzing the network load, the zero cell perspective is not an accurate way of characterization. In what follows, we first define the load of the typical-cell and then, discuss how the load of the zero cell is generally obtained. Then we present our analysis to characterize the load of the typical cell and hence the average network load.

\subsubsection{Average Load of the Typical Cell}
The load of the typical cell in the network can be calculated as:
\begin{align}
\label{TCload}
\rho = \int_\mathcal{A} \frac{w}{C(s)} ds,\end{align}
where $C(s)$ is the rate that a user located at $s$ receives in the typical cell $\mathcal{A}$, calculated using the Shannon formula. The random variable $\rho$ characterizes the load of the cell centered at $x_0$, and depends on the shape and size of $\mathcal{A}$.

The average load of the typical cell is then calculated by taking the expected values of loads for different realizations of the PPP itself: 
$
\bar{\rho} = \mathbb{E}[\rho],
$
\subsubsection{Mean Cell Approximation}
In case of \ac{PV} cells, the average load of the typical cell is generally difficult to evaluate because the shape and size of the typical cells is random.
 However, by using the ergodicity of the \ac{PPP}, the area of the typical cell can be approximated as
$\frac{1}{\lambda}
$ \cite{singh2014joint}.
Then, by assuming the network to be noise-limited, the average load can be approximated using the mean cell approach~\cite{blaszczyszyn2014user}, as: 
\begin{align}
{\bar{\rho}_{MC}} =  \int_{T} \frac{w}{B \lambda \log_2(1 + T)}p(T) dT, 
\end{align}
where the expectation is taken with respect to the \ac{SNR} ($T$) variations averaged over the fast fading, and $p(T) = \frac{-d\mathbb{P}_{C}(T)}{dT}$ is the \ac{pdf} of the \ac{SNR} of the typical user obtained by differentiating the \ac{SNR} coverage probability, $\mathbb{P}_{C}(T)$. Thus, the cell load can be calculated numerically, given the \ac{SNR} distribution.

However in the mean cell approach, as the expectation is taken with respect to the SNR variations of the typical user, it calculates the expected load of the zero cell which is statistically larger than the typical cell~\cite{chiu2013stochastic}. Thus, the mean cell approach, always overestimates the load of the typical cell.
In the next sections, we propose a new approximation, which is both more accurate and more tractable. First, we derive the \ac{CDF} of the cell load of the typical cell using the distribution of its area. Then, we obtain a single-integral-based and a closed-form approximation for the average load of the typical cell.
% \begin{table}[!t]
%     \small
% 	\centering
%     \caption{Notations and System Parameters} % title of Table
% \begin{tabular}{|c | c |}
% 	\hline  Parameter& Value\\
% \hline 
%     	\hline  $\lambda$ & 10-1000 per sq. km. \\  
% 	 \hline {$P_t$}&  {30 dBm}\\ 
% 	 \hline  {$\alpha$}&  {2}\\
%      %\hline  $\alpha_{SNm}$& NLOS SBS path-loss exponent in mm-wave & 8\\
%      \hline  $G_0$& 36 dB \\
%      \hline  $N_0$ & -174 dBm/Hz\\
%      \hline  {$B$}  &  {1 GHz}\\
%      %\hline  $B_{mm}$& mm-wave bandwidth & 1 GHz\\
% 	\hline  \textcolor{blue}{$\theta$}  & \textcolor{blue}{15 degrees}\\
% 	\hline 
% \end{tabular}
% \label{tab:Param}
% \end{table}
\vspace*{-0.6cm}
\subsection{Distribution of the Area of the Typical Cell}
The reduced area of a \ac{PV} cell is defined as~\cite{chiu2013stochastic} :
\begin{align}
s = A/\mathbb{E}[A],
\label{eq:redArea}
\end{align} 
where $A$ is the area of the typical cell, and $\mathbb{E}\left[\cdot\right]$ is the expectation operator. 
The \ac{pdf} of the reduced area of the typical \ac{PV} cell in two dimensions is given by \cite{tanemura2003statistical}:
\begin{align}
f_s(x) = \frac{343}{15}\sqrt{\frac{7}{2\pi}} x^{5/2}\exp\left(-\frac{7}{2}x\right)
\label{eq:pdf}.
\end{align}
Using this, we can obtain the \ac{CDF} of the area as given below.
\begin{lemma}
The \ac{CDF} of the area of the typical \ac{PV} cell for a PPP with intensity $\lambda$ is given by:
\begin{align}
F_A(x) =  \frac{343}{15}\sqrt{\frac{7}{2\pi}} \left(\frac{2}{7}\right)^{7/2}\gamma_{inc}\left(\frac{7\lambda x}{2},\frac{7}{2}\right),
\label{eq:CDF_Def}
\end{align}
where $\gamma_{inc}(\cdot)$ is the lower incomplete gamma function given by: $\gamma_{inc}(x,a) = \int_0^x t^{a-1}\exp(-t)dt$.
\end{lemma}
\begin{proof}
The \ac{CDF} can be easily derived using the relation \eqref{eq:redArea} and integrating \eqref{eq:pdf}.
\end{proof}
\vspace*{-0.4cm}
\subsection{Distribution of the Load of the Typical Cell}
For obtaining the \ac{CDF} of the load of the typical cell, we assume that the shape of \ac{PV} cells is circular. Although in a real \ac{PV} network, almost surely no circular cells occur, our results show that the circular assumption does not greatly deteriorate the derived approximation.
%\footnote{Furthermore, since the load generated at a \ac{BS} by a user is an increasing function of the distance of the user from the \ac{BS}, for the typical \ac{PV} cell of a given area, the circular shape corresponds to the minimum load.} %This can be easily proved by considering a circular coverage region of a given area centered at the origin, and subsequently, observing that a random convex region of any other shape centered at the origin, introduces users that are farther from the center in lieu of users that are relatively nearer. Thus, in essence, our derived mean of the cell load using the circular shape approximation of the typical cell, also serves as a lower bound for the mean load of the typical cell. 
Accordingly, the load of a typical cell \eqref{TCload} with area A is approximated as~\cite{bonald2003wireless}:
\begin{align*}
\rho_{TC}(A) \approx \rho_{AP}(A) = \int_0^{2\pi}\int_0^{\sqrt{\frac{A}{\pi}}} \frac{wr }{B\lambda\log_2\left(1 + \xi r^{-2}\right)} dr d\theta. 
%\label{eq:load_form}
\end{align*} 
\begin{theorem}
\label{theo:CDF}
The distribution of the load of the typical cell, $\rho$ is given by:
\begin{align}
F_{\rho}(l) = F_A\left(\pi \left(\frac{1}{\xi}\exp\left(-\frac{\alpha}{2}\mbox{Ei}^{-1}\left(-\frac{l}{K'}\right)\right)\right)^{\frac{-2}{\alpha}}\right)
\end{align}
where $K' = \frac{ 4w\pi\ln(2)\xi}{\alpha^2B\lambda}\xi^{\frac{2}{\alpha}}$ and the symbol $\text{Ei}^{-1}(x)$ is the inverse of the exponential integral. For the special case of $\alpha = 2$, it is approximated as:
\begin{align}
F_{\rho}(l) \approx 
 F_A\left(\frac{\xi\pi\left(1+\exp\left(-\frac{l}{K_1}\right)\right)}{\exp\left(-\frac{l}{K_1}\right)}\right), %, \quad l \geq 0.5
\label{eq:distribution}
\end{align}
%\hrulefill
%\end{figure}
where $K_1 = \frac{ w\pi\ln(2)\xi}{B\lambda}$.
\end{theorem}
\begin{proof}

According to our assumption of high \ac{SNR} for dense networks, we can approximate $1 + \xi r^{-\alpha}$ as $\xi r^{-\alpha}$. Substituting $\ln\left(\xi r^{-\alpha}\right) = y$, we have:
\begin{align}
\rho_{AP}(A)  &= K' \int_{\ln\left(\xi\left(\pi/A\right)^{\frac{\alpha}{2}}\right)}^{\infty} \frac{\exp(-y)}{y} dy \nonumber \\
&= K'\text{E}_1\left(\frac{2}{\alpha}\ln\left(\xi\left(\frac{\pi}{A}\right)^{\frac{\alpha}{2}}\right)\right),
\label{eq:load_vs_A}
\end{align}
where, $K' = \frac{ 4w\pi\ln(2)\xi}{\alpha^2B\lambda}\xi^{\frac{2}{\alpha}}$ and $\text{E}_1(\cdot)$ is the exponential integral function \cite{barry2000approximation}. The \ac{CDF} is then simply obtained by some algebraic manipulations of the expression $\mathbb{P}(\rho_{AP}(A)\leq l)$.
For the special case of $\alpha = 2$, the \ac{CDF} of the approximated load of the typical cell $\rho_{AP}(A)$ is derived as:
\begin{align*}
\mathbb{P}(\rho_{AP}(A)\leq l)  %\mathbb{P}\left(K_1\text{E}_1\left(\ln\left(\frac{\xi\pi}{A}\right)\right) \leq l\right) \nonumber \\
=\mathbb{P}\left(-\ln\left(\xi\pi A^{-1}\right) \leq \text{Ei}^{-1}\left(-\frac{l}{K_1}\right)\right),
%\label{eq:CDF_Unapprox}
\end{align*}
where, the symbol $\text{Ei}^{-1}(x)$ is given by $\text{Ei}(x) = -\text{E}_1(-x)$. Although an explicit expression for $\text{Ei}^{-1}(.)$ does not exist, Pecina \cite{pecina1986function} provided piece-wise functions to approximate $\text{Ei}^{-1}(x)$ for different ranges of $x$. 
The asymptotic approximation for $\frac{-1}{K_1} \to 0$ is~\cite{pecina1986function}:
\begin{align*}
\text{Ei}^{-1}\left(-\frac{l}{K_1}\right)
 \approx \frac{\exp\left(-\frac{l}{K_1}\right)}{1+\exp\left(-\frac{l}{K_1}\right)}.
 \end{align*}
 
 \begin{figure*}[!h]
 \small
\begin{align}
\bar{\rho}'_{AP} &= \int_0^\infty \lambda\frac{2\pi\ln(2)\xi}{B} F_1(A) \frac{343}{15}\sqrt{\frac{7}{2\pi}}\left(A\lambda\right)^{5/2}\exp\left(-\frac{7}{2}A\lambda\right) dA \tag{12}
\label{Eq:Final_Mean}
\\
\mbox{where,} \hspace{1cm} F_1(A) &= \begin{cases}
\frac{\exp\left(-\ln\left(\frac{\xi\pi}{A}\right)\right)\ln\left(\frac{G_0}{\ln\left(\frac{\xi\pi}{A}\right)} + G_0 + (1-G_0)\beta\left(\ln\left(\frac{\xi\pi}{A}\right)\right)\right)}{G_0 + (1-G_0)\exp\left(\frac{-\ln\left(\frac{\xi\pi}{A}\right)}{1-G_0}\right)}; \qquad &{A \leq \frac{\pi\xi}{\exp(1)}} \nonumber \\
-\gamma - \ln\left(-\ln \left(\frac{\xi\pi}{A}\right)\right) +\left(-\ln\left(\frac{\xi\pi}{A}\right)\right) - \frac{\left(-\ln\left(\frac{\xi\pi}{A}\right)\right)^2}{8}; \qquad &{A > \frac{\pi\xi}{\exp(1)}}
\end{cases}
\end{align}
\hrulefill
\vspace*{-0.4cm}
\end{figure*}
 
 In our analysis, we assume $P_t = 30$ dBm, a noise density of -174 dBm/Hz, and $B = 1$ GHz. The path loss coefficient $K$ is derived from the Umi model for data transmission~\cite{38.900}. As the load $l$ varies from $0 \leq l \leq 1$, for $G_0 = 20$ dB, and $\lambda =  1e-5$ m$^{-2}$, we have $0\leq \frac{l}{K_1} \leq 1e-8$. Thus, our asymptotic approximation is valid. Using this approximation completes the proof.
%Following this approximation, \eqref{eq:CDF_Unapprox} can be written as:
%\begin{align}
%\mathbb{P}(L \leq l) = \mathbb{P}\left(A \leq \frac{\xi\pi\left(1+\exp\left(-\frac{l}{K_1}\right)\right)}{\exp\left(-\frac{l}{K_1}\right)}\right) \nonumber,
%\end{align}
%where $\chi_1 = \frac{l}{K_1}$.
\vspace*{-0.3cm}
\end{proof}
%In Section~\ref{sec:SR} we show that this approximation provides very accurate characterization of the CDF of the network load.
\vspace*{-0.25cm}
\subsection{Proposed Approximations for the Average Load of the Typical Cell}
Using the distribution of \eqref{eq:CDF_Def} and \eqref{eq:load_vs_A} we obtain the following approximation of the average load:
\begin{align}
\bar{\rho}_{AP} = \int_0^\infty  K'\text{E}_1\left(\frac{2}{\alpha}\ln\left(\xi\left(\frac{\pi}{A}\right)^{\frac{\alpha}{2}}\right)\right)f_A(A) dA \label{eq:MEAN_expr},
\end{align}
where $f_A(x) = \frac{dF_A(x)}{dx}$.
Since solving \eqref{eq:MEAN_expr} is tedious, we provide two results to approximate the average load of the typical cell for the special case of $\alpha = 2$\footnote{For dense deployments, the serving \ac{BS} is generally in \ac{LOS} which has a path-loss exponent close to 2 for sub-6GHz~\cite{36.814} and mm-wave~\cite{38.900} transmissions.}. In Section~\ref{sec:SR}, we will highlight the advantage of each approximation.
\subsubsection{\ac{EI} based Approximation}
\begin{theorem}
\label{theo:app1}
The average load of the typical cell can be approximated as (12).
\end{theorem}
\begin{proof}
\stepcounter{equation}
We rely on an approximation of the exponential integral provided by Barry et. al. \cite{barry2000approximation}:
\begin{align}
\text{E}_1 (x) = \frac{\exp(-x)\ln\left(\frac{K_2}{x} + K_2 + (1-K_2)\beta(x)\right)}{K_2 + (1-K_2)\exp\left(\frac{-x}{1-K_2}\right)},
\label{Eq:E1}
\end{align}
where, $K_2 = \exp(-\gamma) = 0.56$, $\beta(x) = 1 - \frac{1}{(h(x) + bx)^2},$ 
\begin{flalign}
h(x) = \frac{1}{1 + x\sqrt{x}} + \frac{0.46 x^{\sqrt{\frac{31}{26}}}}{1 + 0.43 x^{\sqrt{\frac{31}{26}}}}, b \approx 1.04207 \nonumber,
\end{flalign}
and $\gamma$ is the Euler's constant. This interpolated version of the exponential integral provides a good approximation for $1 \leq x \leq 50$. This corresponds to the range \begin{align}
\frac{\pi\xi}{\exp(50)} \leq  A \leq \frac{\pi\xi}{\exp(1)}
\label{eq:A_BOUND}.
\end{align}
For the region of $A$ greater than this range, we use the asymptotic expansion of $E_1(x)$ as:
\begin{align}
E_1(x) = -\gamma - \ln(x) + x - \frac{x^2}{8} + ...
\label{eq:ASYMP}
\end{align}
For practical ranges of cell sizes, the lower bound in \eqref{eq:A_BOUND} always holds (e.g., with $G_0 = 0$ dB, the lower bound on the area is $A \geq 2e-8$).
Now, substituting \eqref{eq:load_vs_A} in \eqref{eq:MEAN_expr} and using \eqref{Eq:E1} and \eqref{eq:ASYMP} to evaluate the integral, completes the proof.
\end{proof}

\subsubsection{closed-form (CF) Approximation}
Substituting $\ln\left(\frac{\xi\pi}{A}\right) = t$, the approximated average load of the typical cell \eqref{eq:MEAN_expr} becomes:
\begin{align}
\bar{\rho}_{AP} = \chi_2 \int_{-\infty}^{\infty} E_1(t) & \exp\left(-\frac{7}{2}t\right)F_2(t) dt,
\label{eq:mean_approx}
\end{align}
where $\chi_2 = \frac{w \pi\ln(2)\xi}{\lambda B}(\xi\pi\lambda)^{3.5}$ and $F_2(t) = \exp\left(-\frac{7\lambda \xi \pi}{2} \exp(-t)\right)$.

Now, a closed-form solution to this integral does not exist. However, in what follows, we derive an approximate closed-form solution for the average load of the typical cell, which we show to be very accurate in Section~\ref{sec:SR}. 
% \begin{figure} 
% \centering
% \includegraphics[width=7cm,height = 3.6cm]{Approx.pdf}
% \caption{Approximation of $F_2(t)$.}
% \label{fig:Approx}
% \vspace*{-0.8cm}
% \end{figure}
\begin{theorem}
The average load of the typical cell is approximated by the closed-form expression:
\label{theo:Closed}
\begin{align}
%\label{CFmeanLoad}
 \bar{\rho}''_{AP} = \chi_2 & \left( I_1(t_2) - I_1(t_1) + (y_1-1)I_2(t_2) - y_1I_2(t_1) \right) \nonumber
\end{align}
where,\\
%\begin{align}
$I_1(x) = \left(\frac{2}{7}\right)^{2} \left(E_1(4.5x) - (1+3.5x)e^{-3.5x}E_1(x)
%\right. \nonumber \\ \left. 
+\right.$ $\left.\left(\frac{7}{9}\right)\exp(-4.5x) \right)$  %\nonumber% \\
% \nonumber
%\end{align}
$ I_2(x) = \frac{2}{7}\left[ E_1\left(4.5x\right) - e^{-3.5x}E_1(x)\right]$, and, $t_1 = -\ln\left(-\frac{2}{7\lambda\xi\pi}\ln\left(0.1\right)\right),  
t_2 = -\ln\left(-\frac{2}{7\lambda\xi\pi}\ln\left(0.9\right)\right) $ and, $y_1 =  \frac{0.9 t_1-0.1 t_2}{t_1-t_2}.$
\end{theorem}
\begin{proof}
We can approximate $F_2(t)$, with a piece-wise defined ramp and step function as follows:
\begin{align}
\tilde{F}_2(t) = \begin{cases}
0 ; \qquad &t \leq t_1, \\
\frac{0.8t}{t_2-t_1} + y_1; &t_1 < t \leq t_2, \qquad \\
1 ;  \qquad & t > t_2 
\end{cases}
\label{eq:approx_F2}
\end{align}
where $t_1$ and $t_2$ are the points corresponding to 10 and 90 percentile values of $F_2(t)$, and $y_1$ is the intercept. %We plot this proposed approximation for $F_2(t)$ in Fig. \ref{fig:Approx}. We observe that $\tilde{F}_2(t)$ provides a good approximation for $F_2(t)$.  
In the next section, we will show that this approximation provides accurate results for the average cell load.
With the approximation of \eqref{eq:approx_F2} in \eqref{eq:mean_approx}, we have:
\begin{align}
\bar{\rho}_{AP} = \chi_2 \left(\int_{t_1}^{t_2}E_1(x)e^{3.5x}\left(\frac{0.8}{t_2 - t_1}x + y_1\right) dx +  \right. \nonumber \\\left.\int_{t_2}^{\infty}E_1(x)e^{-3.5x} dx\right) \nonumber
\end{align}
Geller and Ng \cite{geller1969table} provided closed-form expressions for both of the above integral types, which we employ to obtain the closed-form for the average cell load.
\end{proof}
\vspace*{-0.45cm}
\section{Simulation Results}
\label{sec:SR}
%In this section, we validate the obtained derivations to characterize the \ac{CDF} and the average of the cell load. 
%and the   employ the derived \ac{CDF} of the load to analyze the network stability under various deployment densities and antenna gains. Furthermore, we validate the analytical expression derived for the mean load by comparing it with numerical results of the mean cell approach. 
\subsection{\ac{CDF} of the Load and Stable Fraction of the Network}
To validate the approximation of the \ac{CDF} of the cell load, we compute the stable fraction of the network, which is defined as the fraction of non-overloaded cells. Mathematically, this is the probability that the load of the typical cell is less than 1.
In Fig. \ref{fig:STABLE_FRAC} we compare the stable fraction of the network for a file size of $\sigma = 100$ Mb, and a user arrival rate of $\lambda_U = 100$ km$^{-2}$, obtained with the approximation of the \ac{CDF} derived in Theorem \ref{theo:CDF} and the one computed from Monte-Carlo simulations of the PPP.
This provides dimensioning rules for the operator in terms of the minimum deployment density of \ac{BS}s required to achieve a given stable fraction. For example, with a directive antenna gain of $G_0 = 20$ dB and for a load target of 0.5, the operator must deploy at least $50$ BSs km$^{-2}$.
We also observe that the closed-form \ac{CDF} provides a good approximation of the numerical values, specially for a larger antenna gain $(G_0 = 20 \text{ dB})$. Accordingly, the circular assumption of the cell shape is not detrimental for evaluating the performance of the network. %More interestingly, we observe that increasing $G_0$ has the same effect as increasing the deployment density in terms of network stability. As an example, with $G_0 = 20$ dB, to achieve complete network stability (stable fraction = 1), only 15 BSs km$^{-2}$ are required, whereas, with $G_0 = 10$ dB, the deployment density required to guarantee stability is over 150 BSs km$^{-2}$.
\begin{figure} 
\centering
\includegraphics[width=7cm,height = 3.4cm]{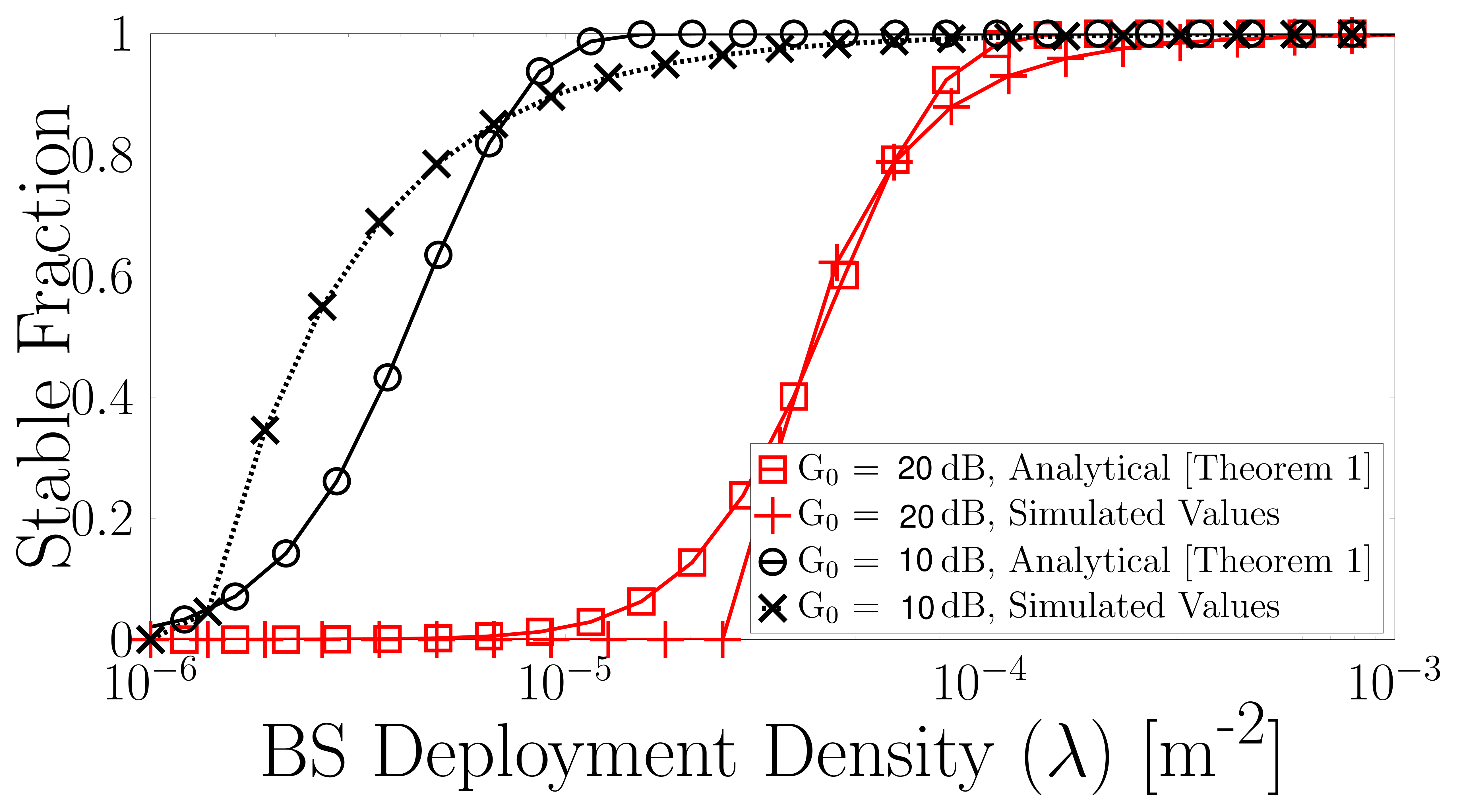}
\caption{Stable fraction of the network.}
\label{fig:STABLE_FRAC}
\vspace*{-0.6cm}
\end{figure}
\vspace*{-0.5cm}
\subsection{Accuracy of the EI Approximation of the Network Load}
In Fig. \ref{fig:MEAN_MIN}, we compare the average load of the typical cell, computed with the EI approximation (Theorem \ref{theo:app1}), the CF expression (Theorem \ref{theo:Closed}), and that obtained using the mean cell approach with the network load calculated using Monte-Carlo simulations. For the Monte-Carlo simulations, we find the average cell load in one realization of the PPP $\phi$, for a given $\lambda$, $\lambda_U$, and $\sigma$ and then we perform the same calculations and average over 1000 PPP realizations.
As seen in the figure, the EI approximation provides a more accurate characterization of the network load than the mean cell approach. The mean cell approach always overestimates the actual load, because, the zero cell is, on average, larger than the typical cell, resulting in higher load. Therefore, from the perspective of an operator, we provide a more realistic, and hence reliable method to characterize the network load and to dimension the network. As an example, for $\sigma = 100$ Mb, and $\lambda_U = 0.01$ users per second, the EI approximation accurately estimates that the operator must deploy 10 \ac{BS} less (120 as compared to 130) than that prescribed by the mean cell approach. 
\vspace*{-0.45cm}
\subsection{Advantages of the CF Approximation of the Network Load}
As we see in Fig.~\ref{fig:MEAN_MIN}, the CF approximation provides the loosest approximation to the network load; however, it does not require numerical evaluation of integrals. Moreover, we see that for higher file sizes $(\sigma = 100\text{ Mb})$ and denser deployment of small cells $(\lambda \geq 1e-4$ m$^{-2})$, even the CF approximation provides an excellent approximation of the network load. 

From a practical perspective, it provides a fast method of accurately estimating the network load without the need of running extensive simulations, which can become infeasible. Particularly, recall that \ac{BS} locations are Poisson distributed. For every realization of \ac{BS} locations, SNR distribution should be computed by drawing all required random variables. Moreover, a dynamic traffic of users arriving in the system, downloading a file, and leaving should be simulated for a sufficient duration to reach the mixing time of the Markov process. This procedure should be repeated for every set of possible parameters if we want to characterize the network load. At last cell overloading may be undetectable as any simulation has a finite duration. For all these reasons, the analytical model presented in this analysis is necessary to provide very quick results and interesting insights to the system.
\begin{figure}[t]
\centering
\includegraphics[width=7cm,height = 3.5cm]{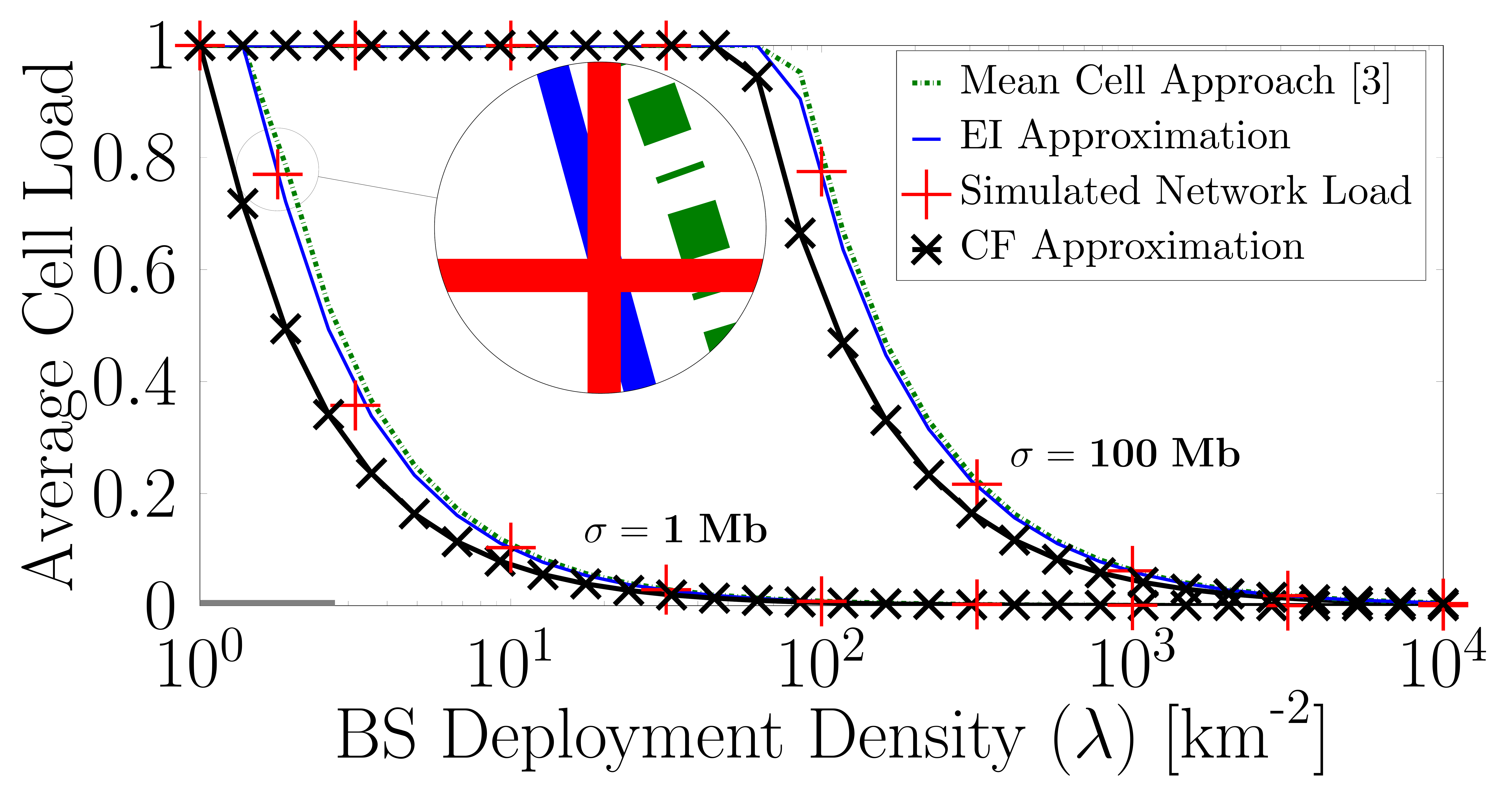}
\caption{Analytical approximation accuracy, $G_0 = 36$ dB.}
\label{fig:MEAN_MIN}
\vspace*{-0.5cm}
\end{figure}
\section{Conclusion}
\label{sec:C}
The realistic assessment of the mobile network performance need to take dynamic traffic into account in order to characterize the network load, which is still an open problem. Towards this end, we have derived a simple approximation for the \ac{CDF} of the cell load of the typical cell in a noise-limited network, which is characterized by high \ac{SINR} and low inter-cell interference. Furthermore, we have obtained a single approximation-based expression and a closed-form expression for the average load of the network by using the distribution of the area of the typical cell. Our derivations present a more realistic characterization of the cell load, as compared to the recently introduced mean cell approach, since we consider the typical cell of the network rather than the zero cell. The analysis provides a tractable and accurate characterization for the cell load that can be utilized, e.g., for evaluating the user throughput and dimensioning 5G networks. However, accurate characterization of the dynamic network load in case of an interference prone network is not straightforward. This will be addressed in a future work.
\bibliography{refer.bib}
\bibliographystyle{IEEEtran}

\begin{acronym}
	\acro{CDF}{cumulative density function}
    \acro{pdf}{probability density function}
    \acro{PPP}{Poisson point process}
    \acro{BS}{base stations}
    \acro{LOS}{line of sight}
    \acro{EI}{Exponential Integral}
    \acro{SINR}{signal to interference plus noise ratio}
    \acro{SNR}{signal to noise ratio}
    \acro{mm-wave}{millimeter wave}
    \acro{PV}{Poisson-Voronoi}
\end{acronym}

\end{document}